\documentclass[11pt]{article}

\usepackage{graphicx}
\usepackage{epsfig}
\usepackage{latexsym}
\usepackage{amsmath}
\usepackage{verbatim}
\usepackage{times}
\usepackage{fullpage}
\usepackage{amsfonts}

\newtheorem{theorem}{Theorem}[section]
\newtheorem{lemma}{Lemma}[section]

\newtheorem{corollary}{Corollary}[section]

\newcommand\QED{\ifhmode\allowbreak\else\nobreak\fi
\quad\nobreak$\Box$\medbreak}
\newcommand{\proofstart}{\par\noindent\bf Proof:\rm\enspace}
\newcommand{\proofend}{\QED\par}
\newenvironment{proof}{\proofstart}{\proofend}

\begin{document}

\title{Stochastic Analysis of a Churn-Tolerant Structured Peer-to-Peer Scheme}

\author{Tim Jacobs \thanks{Department of Computer Science,
Purdue University,
West Lafayette, IN, 47907, USA.} \and Gopal Pandurangan \thanks{Division of Mathematical
Sciences, Nanyang Technological University, Singapore 637371 and Department of Computer Science, Brown University, Providence, RI 02912, USA.  \hbox{E-mail}:~{\tt gopalpandurangan@gmail.com}. Supported in part by the following grants: Nanyang Technological University grant M58110000, Singapore Ministry of Education (MOE) Academic Research Fund (AcRF) Tier 2 grant MOE2010-T2-2-082,
US NSF grant CCF-1023166, and a grant from the US-Israel Binational Science
Foundation (BSF).}}

\date{}

\maketitle

\begin{abstract}
We present and analyze a simple and general scheme to build a   churn (fault)-tolerant structured Peer-to-Peer (P2P) network.
 Our  scheme   shows how to ``convert"  a static network  into a
     dynamic distributed hash table(DHT)-based P2P network such that all the good properties of the static network are guaranteed with high probability (w.h.p).
Applying our  scheme to  a cube-connected cycles  network, for example,  yields a $O(\log N)$
degree connected network, in which  {\em every} search  succeeds in $O(\log N)$ hops w.h.p., using $O(\log N)$ messages, where $N$ is the expected stable network size.  Our scheme has an constant storage overhead (the number of nodes responsible for servicing a data item) and an $O(\log N)$ overhead (messages and time) per insertion and  essentially no overhead
  for deletions. All these  bounds are essentially optimal. While DHT schemes
  with similar  guarantees are already known in the literature,
  this work is new in the following aspects: (1) It presents a rigorous mathematical analysis of the scheme
  under a general stochastic model of churn and shows the above guarantees;
  (2) The theoretical analysis is complemented by a simulation-based analysis
  that validates the asymptotic bounds even in moderately sized networks and also studies performance under changing stable network size;
  (3) The presented scheme seems especially suitable for maintaining dynamic structures
  under churn efficiently. In particular, we show that a  spanning tree of low diameter
  can  be efficiently maintained  in constant time  and logarithmic number of messages per insertion or deletion w.h.p. 
  
  \medskip
  

\end{abstract}

\section{Introduction}
\label{sec:intro}
 Peer-to-Peer (P2P) networks  are highly dynamic: peers enter and leave the network and connections may be added or deleted at any time and thus the topology changes very dynamically.  The independent arrival and departure by a large number of peers  creates a collective effort that is called  as {\em churn}.
Measurement studies of real-world P2P networks ~\cite{SGG02,SW03, SR06,falkner} show that the churn rate is quite high:
nearly 50\% of peers in real-world networks can be replaced within an hour. However,  despite a large churn rate, these studies   show that the
total number of peers in the network is relatively {\em stable}.  The  study by Stutzbach and Rejaie \cite{SR06}  also indicate that P2P networks exhibit a high degree of variance in terms of the session time (the
amount of time spent by a node in the network in a session). They show that the distribution of session times
appear to follow a Weibull or lognormal distribution.

P2P systems  must have efficient and reliable routing in the presence of a dynamically changing
network.  The P2P overlay must exhibit good topological properties (e.g.,
  connectivity, low diameter, low degree, etc.) even if the composition of
  the underlying physical network exhibits significant change.
 Because the system dynamics of these networks are also highly
  asymmetric with only a small number of peers persistent over significant
  time periods, providing churn-tolerance in the presence of mostly
  short-lived peers is essential \cite{SR06}.

 A Distributed Hash Table (DHT) scheme (e.g., \cite{RF+01, SM+01, MNR02, kademlia}) creates a fully decentralized index that maps data
items to peers  and allows a peer to search for a data item
very efficiently (typically logarithmically in the size of the
network) without flooding.   Such systems have been called  {\em
structured} P2P networks, unlike Gnutella, for example, which is
{\em unstructured}.  In unstructured networks, there is no relation
between the file identifier and the peer where it resides.  Structured networks are more difficult to implement than an unstructured networks, mainly due to the fact that it is not easy to maintain a DHT in a highly dynamic setting.
In addition to this, since data is stored typically in some arbitrary node, fault-tolerance to node deletion is essential.
These are some of the reasons why structured P2P networks, despite their efficient search mechanism, have been somewhat  less successful than unstructured networks when it comes to practical deployment. Hence, it is of both  theoretical and practical interest to develop simple and efficient DHT schemes that work provably well under high churn rates.

In this paper, we present and analyze a simple and general scheme to build a   churn-tolerant structured P2P network.
The basic idea behind our scheme is simple.  It is easy to design a static topology with desirable properties such as
connectivity, low degree, low diameter, and an efficient (and local) routing algorithm. Indeed such topologies, e.g., hypercube, butterfly, cube connected cycles, de Bruijn graphs, etc.,  have been studied
extensively in parallel and distributed computing literature.   Our scheme shows how to ``convert"  a static graph   into a fault-tolerant  DHT  network such that all the good properties of the static graph are guaranteed with high probability.
For example, by applying our scheme to  a cube-connected cycles (CCC)  graph  yields a    $O(\log N)$
degree P2P network that has a  $O(\log N)$
latency  (i.e., search time), using $O(\log N)$ messages, with constant storage overhead. Here $N$ is the expected stable network size (cf. Section \ref{sec:model}).  Our bounds are essentially optimal since in our model (cf. Section \ref{sec:model}) if we want all nodes to access all data items with high probability, then it is necessary that the degree be  $\Omega(\log N)$; otherwise there will be
 a non-negligible probability that there would be nodes disconnected from the system.
 In  a dynamic network,  there is the additional challenge of
quantifying the work done by the algorithm to {\em maintain} the
desired properties. An important advantage of the  above protocol
is that it takes $O(\log N)$ overhead (messages and time) per insertion and {\em no} overhead
  for deletions. This is optimal since,  Liben-Nowell et al.
\cite{LBK02} show that $\Omega(\log N)$ work is required to
maintain (even) connectivity in this stochastic model.

In a P2P network it is important to design distributed dynamic algorithms that maintain  fundamental communication primitives such as spanning trees, spanners etc. For example, maintaining a breadth-first search tree is useful for efficient broadcasting, aggregation, and routing. Designing efficient distributed dynamic algorithms is challenging and only few results are known, see e.g., the work of \cite{elkin} that gives a distributed dynamic algorithm for maintaining a spanner.   It is non-trivial
to efficiently maintain even some spanning tree  dynamically --- the trivial
method would be to recompute a spanning tree (e.g., by breadth-first search \cite{peleg}) every time the network changes.
However, this will take $\Theta(D)$ time and $\Theta(|E|)$ messages \cite{peleg}. In contrast, we show how a spanning tree of diameter $O(D)$ (where $D$
is the diameter of the underlying graph) can be maintained by our scheme
in $O(\log N)$  messages and $O(1)$ time per insertion or deletion.  It is not clear how one can efficiently maintain a breadth-first spanning tree or some low-diameter spanning tree in  previous schemes e.g., Chord \cite{BK+03}.

\subsection{Related Work and Comparison}  
The literature on DHT schemes is huge and we 
confine ourselves to those that are relevant to our work.

Our scheme is an improvement in the degree size and message complexity over the network of Saia et
al. \cite{SFG+02}. The structured P2P network described by Saia {\em et
al.}~\cite{SFG+02} has a $O(\log N)$ latency for search, using $O(\log^2 N)$ messages, and
$O(\log^3 N)$ degree. Their fault-tolerant overlay is a butterfly-based expander topology. Their work
guarantees that a large number of data items are available even if a large fraction of
{\em arbitrary} peers are deleted (hence their scheme can tolerate even adversarial deletions unlike ours), under the assumption that, at any time, the number
of peers deleted by an adversary must be smaller than the number of peers joining.
In contrast, our scheme constructs a  $O(\log N)$ latency and $O(\log N)$
degree P2P network that
guarantees that {\em every} search succeeds with high probability (whp)\footnote{Throughout, ``whp" means ``with probability at least $1 - N^{\Omega(1)}$".}
at any time, rather than just a large fraction, under a natural and general
stochastic model ---  the $M/G/\infty$ model \cite{Ross70}. In a $M/G/\infty$ model  the holding (session) times of nodes
can have an {\em arbitrary} distribution, while arrival of nodes is assumed to be Poisson.
(Real-world P2P network measurement studies \cite{SGG02, SR06} have shown that this
is a reasonable statistical model.)
The construction of our overlay is also much
simpler compared to \cite{FS02, SFG+02}. Our scheme also improves significantly on the Warp scheme \cite{Warp}.
Warp guarantees $O(\log N)$ search time, but has a degree of $O(\log^3 N)$.   Our scheme has low maintenance overhead. In particular, node deletions
does not incur any overhead. Multiple nodes can join and leave at the same time (in particular, up to a constant fraction of the total can leave and join at the same time) without any need to change the protocol, and hence our protocol can operate in a highly dynamic setting.

The idea of a general scheme for mapping a static network into a dynamic one 
has appeared before see e.g., \cite{naor, ittai, Manku}.
The work of \cite{cycloid} uses a CCC graph (this is also the graph used to illustrate our general scheme in this paper) to build a structured P2P network. However the above papers do not  present a rigorous analysis of the performance under a  realistic stochastic model. Furthermore, to the best of our knowledge, none of the previous works, address the problem of efficient maintenance of spanning substructures 
 under churn. 
 
There has been other  works  on building fault-tolerant   DHTs under  different deletion models --- adversarial deletions and stochastic deletions. For example, the works of
of \cite{kuhn-schmid, cuckoo, shell} deals with  adversarial churn and gives techniques
to handle worst-case joins and leaves.  Fiat and Saia \cite{FS02} proposed a DHT  network that is robust against adversarial deletions (i.e., an adversary can choose which nodes to fail). In this model some small fraction
of the non-failed nodes would be denied from accessing some of the data items. While this solution is more general than our model it has some disadvantages: (1) It is not clear whether the system can guarantee its bounds when nodes leave and join dynamically; (2) the message   complexity is large --- $O(\log^3 N)$ and so is the network degree. Moreover their construction is very complicated which can increase the likelihood of error in implementation and decrease the possibility of practical use. In a subsequent paper Saia et al. \cite{SFG+02} address the first problem and give a scheme with  $O(\log N)$ time for search, using $O(\log^2 N)$ messages, and
$O(\log^3 N)$ degree. Datar \cite{Datar06} gives a scheme based on the multibutterfly network that improves on the scheme of Fiat and Saia \cite{FS02} under the adversarial deletion model. Naor and Weider \cite{MW03} describe a simple DHT scheme that is robust under the following simple random deletion model --- each node can fail independently with probability $p$. They show that their scheme can guarantee logarithmic degree, search time, and message complexity if $p$ is sufficiently small.
In contrast, our scheme is  simpler than  \cite{MW03} and works under a more realistic stochastic deletion model
(even a large constant fraction of nodes can get deleted in our model) and guarantees the same (essentially optimal) performance bounds. Also our scheme requires no maintenance overhead under deletions unlike the scheme of \cite{MW03}.
Hildrum and Kubiatowicz \cite{HK03} describe how to modify two popular DHTs, Pastry \cite{Pastry} and Tapestry \cite{ZKJ01} to tolerate random deletions.
 Finally, we point out that several DHT schemes (e.g., \cite{SM+01,RF+01,Koorde}) have been  shown to be robust under the simple random deletion model mentioned above.
 
 Recently there has been interest in designing distributed algorithms for fundamental distributed computing tasks
 such as agreement and leader election. The work of Kapron et al. \cite{Saia} addresses this in a static setting (where the graph is fixed) under a byzantine adversary. The recent work of Augustine et al. \cite{soda12}
 is the first work that addresses the agreement problem in a dynamic P2P network under an adversarial
 churn model where the churn rates  can be very large, up to linear in the number of nodes in the network.


\section{The P2P Scheme}
\label{sec:protocol}
 We will show how to build a P2P  network  $G = (V_G,E_G)$ of expected stable size $N = |V_G|$ (defined precisely in Section \ref{sec:model}). Let $H = (V_H,E_H)$   be the static (``template")  graph
that will be used to build $G$. (We will later show how the network can
dynamically be made to adapt to a changing network size. Note that stable means
that the total network size is more or less remains the same, up to constant factors.)
 We will use the term {\em node} to
  denote a node (peer) of $G$ and the term
  {\em vertex}  to denote a vertex of $H$.

Although, in principle,  any graph can be taken as a template (or ``backbone") graph, for the purposes of
constructing efficient P2P networks   it is desirable that $H$ has certain
properties such as connectivity, regularity,  recursive structure, constant degree, logarithmic diameter,  and a simple and efficient (local) routing scheme.
Good candidates for $H$ are hypercube network  and its variants (butterfly,  Bene\v{s } network   and cube connected cycles),   de Bruijn graph etc.
Henceforth, the following assumptions will be made with respect to $H$:
\begin{enumerate}
\item $H$ has diameter $D$ and maximum degree $\Delta$.
\item $H$ has a  local and efficient routing scheme ${\cal R}$ that can route between any two nodes in $O(D)$ time using $O(D)$ messages,
 where $D$ is the diameter of $H$. Specifically it will be required that $H$  has a vertex labeling scheme that enables
 shortest path routing with low memory overhead (see e.g., \cite{PG} for a survey on such routing schemes).
 In such a routing scheme,  vertex labels are assigned in such a way that every  vertex $v$,  given the destination address $u$,  can decide locally (based solely on the address of $u$) the  outgoing edge of $v$ that (eventually) leads to $u$ by using only a routing table of size at most $O(\Delta)$ entries per node.
  (Each entry of the routing table will specify which outgoing edge to take for a given destination $u$.)  A well-known example of such a scheme is the {\em bit-fixing routing scheme} in a hypercube (and  its variants)  \cite{leighton}.
\end{enumerate}

Given the above assumptions, our scheme builds a DHT-based P2P network (with expected stable size $N$) with the following properties:
\begin{itemize}
\item The  degree of a node and its routing table size  is bounded by
$O(\Delta \log N)$ w.h.p. (cf. Theorem \ref{th:degree})
\item  At any time, the network is connected and has a diameter of $O(D)$ w.h.p. (cf. Theorem \ref{th:conn})
\item {\em Every} search will succeed in $O(D)$ time w.h.p and will use $O(D)$ messages. (cf. Theorem \ref{th:search})
\item The
time and message overheads for a node to join the network are $O(D)$ and $O(D + \Delta \log N)$ respectively w.h.p.  (cf. Theorem \ref{th:overhead})
\item Number of nodes responsible for servicing a data item is $O(1)$.  
\end{itemize}


  Throughout, we will illustrate by taking $H$ to be a {\em cube connected cycle (CCC)} network. Our scheme can be adapted to other similar types of graphs.   The  $r$-dimensional CCC is constructed from
the $r$-dimensional hypercube by replacing each vertex of the hypercube with a cycle of $r$ vertices in the CCC. The $i$th
dimension edge incident to a vertex of the hypercube is then connected to the $i$th vertex of the corresponding cycle of the CCC \cite{leighton}. In a CCC, the label of a vertex is represented by a pair $<w,i>$ where
  $i$ is the position of the vertex within its cycle and $w$ is the label of the vertex in the hypercube that corresponds
 to the cycle. Two vertices $<w,i>$ and $<w',i'>$ are linked by an edge in the CCC if and only if either
  (1) $w = w'$ and $i-i' \equiv \pm 1 \mod r $  or (2) $i = i'$ and $w$ differs from $w'$ in precisely the $i$th bit.
    Edges of the first type are called cycle edges, while edges of the second type are referred to as hypercube edges. A CCC graph of $N$ vertices has diameter $O(\log N)$ and  each vertex has degree 3.
A CCC  has an efficient routing scheme, namely the bit-fixing routing scheme \cite{leighton}  that can route in $O(\log N)$ time using
$O(\log N)$ messages using only routing tables of size $O(\log N)$.  In this scheme, to route a message between two vertices with vertex labels $<u,i>$ and $<v,j>$, the bits of $u$ are successively transformed (say, from the first to the last) to match $v$.  The message is routed between one dimension to the next using the hypercube edges, while the cycle edges are used to bring the message to the vertex of the cycle with the appropriate dimension.

   A node
 $x$  in $G$ has a label called
the {\em node-id}  which corresponds to a
 vertex label of $H$.
  We will
  choose the size of $H$ to be $S= |V_H| = \frac{N}{\alpha\log N}$, where $\alpha$ is a (suitably large) fixed constant.  (Throughout we will assume logarithm to the base 2. We will omit floors and ceilings, assuming that quantity in question is rounded to the nearest integer.)
  Node-ids are assigned randomly by sampling from all possible vertex labels of $H$.
Specifically,  if $H$ is a CCC, the  node-id of a node is obtained as follows: toss a fair coin (has a equal probability of getting
a 0 or 1) $r = \log(N/\alpha\log^2 N)$ times independently
and obtain
a $r$-bit random bit string $\sigma_r$ ($r$ is the dimension of the CCC).  Also sample  a random number from
  $1$ to $r$ and call it $i$. Then the  node-id  of the node is  $<\sigma_r, i>$.
 We say that the node {\em covers} the  vertex having the label corresponding to its node-id.
   There is an
edge between two nodes   with node-ids $x$ and $y$  if there is an edge in $H$ between $x$ and $y$ or  if $x =
y$. (Thus note that nodes that share the same vertex label will form a clique.) We call a vertex in $H$ to be {\em occupied} if there is a  node in the network
(i.e., a live peer) which covers this vertex; otherwise we call it to be a {\em hole}.

\medskip


\noindent {\bf Joining and Leaving the Network.}
 A  node   (say $v$) that wants to join the network chooses its  node-id as explained above. We assume that $N$ (the expected stable network size) or an estimate  of $N$  (a constant factor estimate is sufficient) is known to all joining nodes.  Because of the numbering
scheme, $v$ can locally determine the node-ids of its (potential)  neighbors without any global knowledge.
$v$'s neighbors in the P2P networks are the nodes that cover the above determined node-ids.
To join the network, $v$ contacts any one of the nodes in the network (such entry points are provided by an external mechanism).   $v$   can then make use of an efficient routing scheme ${\cal R}$ of $H$ to   find its neighbors (i.e., their IP addresses) and joins by connecting to them. If $H$ to be a CCC,   ${\cal R}$ can be the standard {\em bit-fixing routing scheme} mentioned earlier.

A node can  leave the network at any time; the node's data is transferred
to a randomly chosen node with the same vertex label (note that all such nodes
are neighbors of the leaving node). We show later that such a node will always exist w.h.p in our model.

\medskip

\noindent {\bf Search (Look-up) Scheme.}
Searches are handled by a DHT scheme, similar to other DHT schemes such as Chord
\cite{BK+03}.  The data (or key) is hashed to a random vertex label in the same fashion as was done for choosing the node-id of a vertex.
Data is inserted to a randomly chosen node having this label as
its node-id.  Search for
this data is thus directed to some node (say $u$) having its node-id equal to the data's
hashed value. The data will be stored in $u$ or any one
of the neighbors of $u$ that share the same vertex-label.  Since all nodes sharing a node-id are connected  to each
other (forming a clique), search will succeed even if only one node covering this vertex is live in the
network (this node will have the data). Search is performed by invoking the bit-fixing routing scheme  as illustrated below by an example.
Suppose a node with node-id $x$ wants to search for a data item hashed to a number $t$ ( $1 \leq t \leq S$).
Let the route given by the bit-fixing routing  scheme from $x$  to $t$ be $<x, u_1,  u_2, \dots, t>$.
Then $x$ will send a message to a neighbor node which
covers $u_1$ which in turn will forward  to its neighbor node which covers $u_2$ and so on.


\section{Analysis of the P2P Scheme}
\label{sec:analysis}
 We analyze various
 network parameters -- routing table size (i.e., degree), connectivity and diameter,
 maintenance overhead for joins, and  the complexity for doing search.
We  first describe the
 stochastic model used in our analysis.

\subsection{Stochastic Model}
\label{sec:model}
In evaluating the performance of our protocol we focus on the long term behavior of the
system in which nodes arrive and depart in an uncoordinated, and unpredictable fashion.
We model this setting by a stochastic  continuous-time process: the arrival of new
nodes is modeled by Poisson distribution with rate $\lambda$, and the duration of time
a node stays connected to the network is independently determined by an {\em arbitrary}
distribution $G$ with mean $1/\mu$. This is also called the $M/G/\infty$ model in
queuing theory \cite{Ross70}.   (This is  more  realistic than the less general  $M/M/\infty$
 used in
\cite{PRU01} to model P2P networks.) Measurement studies of real P2P systems \cite{SGG02, SW03,SR06}
indicate that the above model approximates real-life data reasonably well, especially
since the holding time distribution is arbitrary (in particular the study in \cite{SR06} actually indicates that
the holding times may follow Weibull or lognormal distributions).

Let $G_t$ be the network at time $t$ ($G_0$ has no vertices). We are interested in
analyzing the evolution in time of the stochastic process ${\cal G}=(G_t)_{t\geq 0}$.
Since the evolution of ${\cal G}$ depends only  ${\lambda/\mu}$ we can assume w.l.o.g.
that $\lambda=1$. To demonstrate the relation between these parameters and the network
size, we use $N={\lambda/\mu}$ throughout the analysis. We justify this notation by
showing that the number of nodes in the network rapidly converges to $N$ which we call the {\em
expected stable network size} (or simply, stable network size). We use the
notation $G_t=(V_t,E_t)$ be the network at time $t$. 

Throughout our analysis we use the Chernoff bounds for the binomial and the Poisson
distributions. Let the random variable $X$ denote the sum of $n$ independent and
identically distributed Bernoulli random variables each having a probability $p$ of
success.  Then, $X$ is binomially distributed with $\mu = np$. We have the following
Chernoff bounds~\cite{AS92}: 
For $ 0 < \delta < 1$: $\Pr(X > (1 + \delta)\mu ) \leq
e^{-\mu \delta^2/3}$ and $\Pr(X < (1 - \delta)\mu ) \leq e^{-\mu \delta^2/2}$. We have
identical bounds even when $X$ is a Poisson random variable with parameter $\mu$
\cite{AS92}.

\subsection{Network Size}

The stable network size can be computed using the fact that the joining and leaving of nodes
follow the $M/G/\infty$ queuing model, a standard model in queuing theory (see e.g., \cite{Ross70}).
The following theorem characterizes the stable network size (i.e., the network size after the system stabilizes
according to the queuing model) and is a consequence of the fact
that the number of nodes at any time $t$ is a Poisson distribution (this is true even if
the holding times follow an arbitrary distribution) \cite[pages 18-19] {Ross70};
applying the Chernoff bound for the Poisson distribution gives the high probability
result.

\begin{theorem}
[Stable Network Size] \label{th:size} Assume that the arrival of new nodes
nodes is modeled by Poisson distribution with rate $\lambda = 1$, and the duration of time
a node stays connected to the network is independently determined by an {\em arbitrary}
distribution $G'$ with mean $1/\mu$. Let $N={\lambda/\mu}$.  If $\frac{t}{N} \to \infty$, then $E[|V_t|]=N$, and
 $|V_t|=N \pm \Theta(\sqrt{N\log N})$ with probability at least $1 - 1/N^2$.
\end{theorem}

\begin{proof} Consider a node that arrived at time $x \leq t$. The probability
 that the node is still in the network at time $t$
 is $1- G'(t-x)$. Let $p(t)$ be the
 probability that a random node that arrives
 during the interval $[0,t]$ is still in the
 network at time $t$, then (since in a Poisson
 process the arrival time of a random element is
uniform in $[0,t]$),
$$p(t)={1\over t}\int_{0}^{t} (1-G'(t-x))d \tau ={1\over t}\int_{0}^{t} (1-G'(x))d \tau .$$

 Our process is similar to an infinite server
 Poisson queue. Thus, the number of nodes in the
 graph at time $t$ has a Poisson distribution with
expectation $t p(t)$ (see \cite[pages 18-19] {Ross70}).

Thus, $E[|V_t|] = t p(t)$.
When ${t/N}\to \infty$, $E[|V_t|]=N$.

 We can now use a tail bound for the Poisson
 distribution~\cite[page 239]{AS92} to show that
 \begin{eqnarray*}
 Pr \left(||V_t|-E[|V_t|]|\leq  \sqrt{bN\log N}
 \right ) \geq 1-{1/{N^2}}
 \end{eqnarray*}
 for a suitably chosen constant $b$.
\end{proof}

The above theorem assumed that the ratio
 $N=\lambda/\mu$ was fixed during the
 interval $[0,t]$. We can derive a similar result
  for the case in which the ratio changes
 to $N'=\lambda'/\mu'$ at time $\tau$.

 \begin{corollary}
 \label{cor:size}
Suppose that the ratio of  between arrival and departure rates in the network changed
at time $\tau$ from $N$ to $N'$. Suppose that there were
 $M$ nodes in the network at time $\tau$, then if
 ${{t-\tau}\over {N'}}\to \infty$  w.h.p. $G_t$
has $N'+o(N')$ nodes.
 \end{corollary}

\subsection{Network Degree}

\begin{theorem}
\label{th:degree}
 [Degree]
 At any time $t$ such that $t/N
\rightarrow \infty$, the  degree of a node and the routing table size is bounded by
$O(\Delta \log N)$ w.h.p., where $\Delta$ is the maximum degree of $H$, the template graph.  (If $H$ is a CCC, then
the degree is  $O(\log N)$.)
\end{theorem}


\begin{proof}
We first show that the number of nodes covering a given vertex is $\Theta(\log N)$ w.h.p.

 Let  $Y^v_j$ be the
indicator (0-1) random variable for the event that any given (live) node $j$ covers a given vertex $v \in H$ , i.e.,
\begin{align*}
Y^v_j = & \,\, 1 \mbox{ if node } j \mbox{ covers vertex } v \\
= & \,\, 0 \mbox { otherwise }
\end{align*}
Then,
  $$\Pr(Y^v_j = 1)  = \frac{1}{S}$$
where $S =N/(\alpha\log N)$, is the size of $H$ ($\alpha$ is a constant --- cf. Section \ref{sec:protocol}). The above probability is due to the fact that
each node chooses a vertex label with probability uniformly in the size of the template graph, i.e., $S$. 
Let $Y^v = \sum_{ j \in V_t} Y_j$ be the random variable, denoting the total
number of (live) nodes covering vertex $v$.

By Theorem \ref{th:size}, when the network is stable (i.e., $t/N \rightarrow \infty$), the number of live nodes in $G_t$ is at least $N- o(N)$  with probability at least $1 - 1/N^2$. Hence,  by linearity of expectation, the
  expected number of live nodes covering vertex $v$ in $G_t$ is at least
  $$E[Y^v] = (N - o(N)) \frac{1}{S} (1 - 1/N^2) = \alpha \log N - o(\log N) \geq 8\log N$$
  for $\alpha$ sufficiently large.

 We note that the $Y^v$ is a sum of  independent $0-1$ random variables $Y^v_j$s, and hence we can
 apply the Chernoff bound:
   $$\Pr(Y^v < \log N) = \Pr(Y^v < 8\log N (1 - 7/8)) = e^{-8\log N (7/8)^2 (1/2)} < 1/N^{2}. $$
Hence, with probability at least $1 - 1/N^{2}$, the number of nodes covering vertex $v$ is at least $\log N$, as claimed. Since there are a total of $N/(\alpha \log N)$ vertices in $H$, by union bound \cite{MU}, with probability at least $1 - 1/N$,
the above property is satisfied for all vertices.

The maximum  degree of a vertex  in the template graph is $\Delta$ and since each vertex
covered by $\Theta(\log N)$ nodes w.h.p.,  the degree of a node
in the P2P network is $O(\Delta \log N)$.  
  The bound on the routing table size follows from the fact that $H$ admits an routing scheme ${\cal R}$ that has a routing table size
of $O(\Delta)$ entries per node.
\end{proof}

\subsection{Fault-tolerance and Search}
\label{sec:fault-t}

We show that every query succeeds w.h.p at any time (after a short initial period). We
show this by first proving that every vertex is occupied w.h.p which ensures that
 queries that (logically) map to this vertex value can be serviced by some live node covering this vertex.
 This fact along with the way  edges are constructed  in the P2P network will
  show that a search  will succeed w.h.p for every search.

The following theorem shows there is no hole w.h.p. Recall that we call a vertex as a hole (see
Section 2) means that there is no node (i.e., a live peer) in the network that
covers this vertex. 

\begin{lemma} [Occupancy of Vertices]
\label{th:nohole}  At any time $t$, such that $ t/N
\rightarrow \infty$, w.h.p. every vertex of $H$ is occupied.
\end{lemma}

\begin{proof}
When $t/N \rightarrow \infty$, by Theorem \ref{th:size},  the number of nodes in the network
is at least $N - o(N)$, with probability at least $1 - 1/N^2$.  Each node has the same (uniform) probability to occupy each of the $S$ vertices of $H$, independently of the other nodes.
 Thus, at any time $t$ such that  $t/N \rightarrow \infty$, the probability that a vertex $v$ is not covered  is at most
 $$\Pr(v \mbox{ is not covered } | \mbox{ at least } N - o(N) \mbox{ nodes are in } G_t ) + \Pr (\mbox{ less than } N- o(N) \mbox{ nodes are in } G_t)$$
(In the above, we make use of the inequality $\Pr(A) \leq \Pr(A|B)  + \Pr(\bar{B})$, for any events $A$ and $B$. Specifically, here $A$ denotes the event ``vertex $v$ is not covered" and $B$ denotes the event ``at least  $N - o(N)$ nodes are in $G_t$".)
$$ = \left(1 - \frac{1}{S}\right)^{(N -
o(N))} (1 - 1/N^{2})  + 1/N^2$$ 

Thus the  probability that a vertex $v$ is not covered is
$$ \leq \left(1 - \frac{\alpha \log N}{N}\right)^{(N -
o(N))} (1 - 1/N^2) + 1/N^2$$
$$  \leq e^{-\alpha\log N} < 1/N^2 + 1/N^2 = 2/N^2$$
by our choice of $S = N/(\alpha \log N)$ ($\alpha > 2$).

 Applying the union bound \cite{MU}, the probability that
no vertex is unoccupied is at most $2/N$.
\end{proof}

The following theorem on the success probability of a search query is a consequence of the
previous theorem and the way nodes link to each other. Note that 
we assume that one time unit is  taken for sending a message across an edge (i.e., one hop).

\begin{theorem} [Search]
\label{th:search}
For any time $t$, such that $t/N
\rightarrow \infty$, w.h.p. any search  query will be successful. The
time (number of hops) needed is $O(D)$ w.h.p., where $D$ is the diameter of $H$.
\end{theorem}

 \begin{proof}
  Consider a
search query emanating  at time $t$ from the node with node-id $x$  for a node covering
a node-id $y$ (the hash value of the data).   This search will be
 successful if there is a path in $G$  to one of the nodes covering $y$. In terms of the template graph $H$,
 consider the  path from $x$  to $y$  given by the  routing scheme ${\cal R}$ of length $O(D)$.  (If $H$ is a CCC, then
  the ${\cal R}$ is the bit-fixing scheme and $D$ is $O(\log N)$.) This path goes through a sequence
 of vertices in $H$.  From Lemma \ref{th:nohole}, it follows (via union bound)
 that  every vertex of $H$ is occupied w.h.p for a time interval $O(D)$ (starting from time $t$.)
 Thus during this time interval,  every vertex in $G$ is covered by
 some (live) node in the network. From our construction of $G$
 there is an edge between any node covering a vertex to any node covering the
 neighbor of the vertex. Thus, w.h.p the query will take $O(D)$ time.
 \end{proof}

\noindent The above theorem also implies the following result on the connectivity and
diameter of the network.

\begin{corollary} [Connectivity and Diameter]
\label{th:conn}
 For any time $t$, such that $ t/N
\rightarrow \infty$,  the network is connected and has a diameter of $O(D)$ w.h.p.
\end{corollary}

\begin{theorem}[Overhead of Joining]
\label{th:overhead}
For any time $t$, such that $ t/N \rightarrow \infty$,  the
time and message overheads for a node to join the network are respectively $O(D)$ and $O(D + \Delta \log N)$ w.h.p.
\end{theorem}

\begin{proof}
 An incoming node has to locate a node in the network
with the same node-id; then it can find all of its neighbors in $O(1)$ time and  $O(\log N)$ messages. Finding such a node (starting from some entry point node) takes $O(D)$ time (Theorem \ref{th:search}).  Hence the
total time needed to find all  neighbors is  $O(D )$. The total number of messages needed is $O(D + \Delta \log N)$  w.h.p., since $D$ messages are needed for routing (to find a node of same id) and a routing table updates of size $O(\Delta \log N)$ has to be done in total (for the new node as well as the neighbors of the new node).
\end{proof}


\section{Dynamic Maintenance of Spanning Tree of Low Diameter}
The scheme admits a simple local algorithm to dynamically maintain
a  spanning tree whose diameter is almost optimal, i.e., essentially the same
as the underlying template graph, i.e., $D$. Note that
the diameter of the P2P network is $O(D)$.  Let $G_t$ be the network at time
$t$. The goal is to compute a spanning tree of $G_t$, denoted by $T(G_t)$
of diameter $O(D)$ efficiently. 

The P2P network constructed by our scheme admits a very simple and efficient
algorithm. Let $T(G_{t})$ be the spanning tree of diameter $O(D)$ at some time $t$, such that $ t/N \rightarrow \infty$. We will first describe how $T(G_t)$ is constructed at some time $t$ and then describe how it is maintained under insertions and deletions at any time $t > t'$. With a very small probability, one may have to  construct
the spanning tree from scratch at any time $t$, as discussed below.

 Let $S(\ell)$ be the set of nodes that share the same-vertex
label $\ell$. Construct a breadth-first tree $T_H$ on the template graph $H$.  Choose a (distinguished) node $u(\ell) \in S(\ell)$ --- we call $u(\ell)$  the  leader node
of the set of nodes belonging to $S(\ell)$. The tree $T(G_t)$ is constructed as follows.
Connect the leader nodes of the respective vertex labels as they are connected
in the breadth-first tree $T_H$. Make all non-leader nodes of $S(\ell)$ the children
of the leader node $u(\ell)$. Note that non-leader nodes will all be leaves in $T(G_t)$.

The tree is maintained  as follows:

\noindent {\bf Insertion:} Let a node $v$ is inserted. Let it have vertex label $\ell$. Then the node  is added as a child of the leader node of $S(\ell)$. (Note that a leader exists
w.h.p by Lemma \ref{th:nohole}.) The time and message complexity is $O(1)$ per insertion whp.

\noindent {\bf Deletion:} Let a node $w$ be deleted. There are two cases. If $w$ is a non-leader node,
it is simply removed. Note that this does not disconnect the tree as this will be a leaf node. On the other hand, let $w$ be a leader node and let the vertex label of $w$ be $\ell$. Then
 $w$ is deleted
and in its place another (non-leader) node, say $x$, belonging to $S(\ell)$  (i.e., nodes that have the same vertex label $\ell$)  is elected as leader. By our tree construction, $x$ will be a leaf child (again
such a node will exist w.h.p by Lemma \ref{th:nohole}). Thus the rest of the tree is not affected. Also, by our construction, $x$ will
have an edge to $w$'s parent node and all its other children. Thus connectivity 
is preserved and diameter is still $O(D)$.  The  message complexity
is  $O(\log N)$ per deletion (since only so many nodes are affected) and the time complexity is $O(1)$ (all bounds hold whp). Note that
leader election itself can be done in $O(1)$ rounds and $O(\log N)$ message complexity as the set $S(\ell)$ forms a clique.

There is a small probability that the above algorithm will fail, e.g.,
deleting a node, leaves the corresponding vertex unoccupied (i.e., a hole).
In such a case, one has to reconstruct the tree from scratch as discussed first.

Hence we can state the following theorem.
\begin{theorem}
For any time $t$, such that $ t/N \rightarrow \infty$,
a spanning tree of diameter $O(D)$ (where $D$
is the diameter of the underlying graph) can be maintained using $O(\log N)$  messages  w.h.p. and $O(1)$ time w.h.p. per insertion or deletion.
\end{theorem}

\section{Handling Change in Stable Network Size}
\label{sec:changing}
The performance of our scheme depends on the stability of the network. It is easy to see that our scheme can easily tolerate changes up to constant factors (thus, as mentioned earlier, it is enough to have an estimate of $N$ up to some constant factor). However,
bad events,  such  as the network size
  drastically getting reduced, possibly even leading to the  network getting disconnected,  can happen, but
 with minuscule probability in our model. In case such events happen (which will eventually happen with
 probability 1 if the system runs forever)
remedial measures  can be taken such as resorting to an external mechanism to connect
the network again (if the network gets disconnected) or
 rejecting new connections (if the size exceeds very much)  till the situation self-corrects
 itself. Our analysis can be extended to handle such situations.

We now discuss how the scheme can be modified to accommodate gradual changes in stable network size.
As shown in Corollary \ref{cor:size}, if the ratio between the   arrival and departure rates in the network change, then this leads to a new expected stable network size.  Suppose the new stable size  is  one-half  of the original network size.
How can the network adapt to this changed size? Assume that $H$ is a hypercube (similar argument will work for CCC and other related networks).  All that is required is to  reduce  $S$ (the size of $H$) by a factor of 2. This can be done easily in a local manner. Each node will simply reduce its dimension  by 1. This can be accomplished by
dropping the last bit in the node-id.  The hash values of data are also altered in the same way. It is easy to see that because of the recursive nature of construction of the hypercube (i.e., a hypercube of dimension $r$ can be constructed from two hypercubes of dimension $r-1$), reducing the dimension will require only $O(1)$ overhead per node. To illustrate, consider two nodes with node-ids $<x_1,\dots, x_r,0>$ and
$<x_1,\dots, x_r,1>$. Dropping the last bit, will make {\em both}  these nodes to cover the vertex with label $<x_i, \dots, x_r>$. Data that were originally serviced by either of these will now be serviced by both of them.
On the other hand, if the stable network size increases by a factor of two, then each node will increase its dimension by one, by adding one more random bit to its node-id (cf. Section \ref{sec:protocol}).   To illustrate, consider the set of nodes with the same node-id $<x_1,\dots, x_r>$.  Randomly adding one more bit (last bit), will make on the average half of the nodes
in this set to cover the vertex with label $<x_1, \dots, x_r,0>$ and the other half to cover $<x_1,\dots,x_r,1>$.
The data that are serviced by these nodes also get hashed to the same node-ids.
It is not difficult to show that the above transformation preserves all the properties of the scheme, namely network degree, number of hops needed for search, fault-tolerance, connectivity, and diameter.

\section{Simulation Results}
Our P2P scheme (\ref{sec:protocol}) guarantees certain properties to hold within the network, such as logarithmic diameter, logarithmic searching, etc., with high probability.  The theoretical results proved earlier are asymptotic, i.e., shows that the above properties hold when $\frac{t}{N} \to \infty$. We  present a simulation of the scheme   to get a better picture of how the network will react in practice. Thus it is also of interest to see performance data measured from simulations. 

\subsection{The Simulator}
The simulator is written in Java  to mimic a network that runs for some time $t$ with stable network size $N$. The simulation of all nodes is done in serial.  
The network loops for $t$ cycles, each adding a sequence of new peers, then removing peers who have stayed for their predetermined length. Every so many cycles, the network is inspected to determine varying statistics, e.g. diameter, average degree. 

The network of peers(or nodes) is built upon a "backbone" graph. 
Our scheme allows for a wide variety of "backbone" graphs; the simulator presented here uses a Cube Connected Cycle (CCC) graph. The network is set up with an expected stable network size, $N$, and a runtime length, $t$. 
The dimension $r$ of the CCC graph to build the network on is then determined by 
$r = \log(N/\log^2 N)$ (cf. Section \ref{sec:protocol}). 
The simulator is then looped for $t$ cycles, each cycle composed of: removal of nodes whose time has expired, addition of new nodes, and calculation of network statistics.
The removal of nodes is simple; each is assigned a session length when they arrive based on the probability distribution which governs how long the node stays in the network. This length is decreased every cycle until it reaches zero, when it is removed according to our scheme.

Nodes arrive based on a Poisson distribution with rate $\lambda$. 
This is achieved by each cycle sampling a Poisson random variable with rate $\lambda$ and adding that many nodes to the graph, serially. 
Each node's session length, $1/\mu$, 
is sampled then from an arbitrary distribution, with mean $1/\mu$. 
Based on real world statistics in
Stutzbach, et al. \cite{sampling, SR06}, Weibull, log-normal, and exponential distributions fit well to mirror actual peer session lengths.
For most simulations, the session length was taken from random variable with Weibull distribution, with shape parameter $k = .59$ (based on \cite{sampling}) 
and varying the scale parameter $\lambda$ such that the mean of the random variable will be $1/\mu$.
The Poisson variable rate $\lambda$ and Weibull mean $1/\mu$ are then chosen so
$\lambda/\mu = N$. The node is then added to the network as described in Section \ref{sec:protocol}.

The simulator keeps track of the basic network statistics: diameter, average degree, as well as
those of interest to this specific network construction: vertex coverage, i.e. the percentage of 
vertices in the "backbone" network that are covered by network nodes; average coverage, i.e. 
the average number of nodes covering a vertex in the "backbone"; and random path length, i.e.
the average path length through the network over $\log(N)$ paths.

\subsection{Results}

We now present our simulation results on various network parameters, namely coverage, connectivity,
diameter, and how the network reacts to changes in network size.

\begin{figure}[ht]
\begin{minipage}[b]{3in}
\centering
\includegraphics[width=3in]{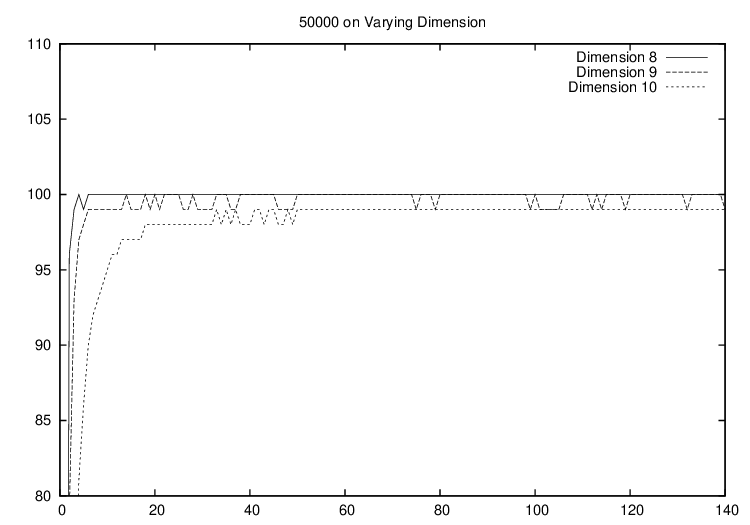}
\caption{Coverage}
\label{fig:figure1}
\end{minipage}
\hspace{0.5in}
\begin{minipage}[b]{3in}
\centering
\includegraphics[width=3in]{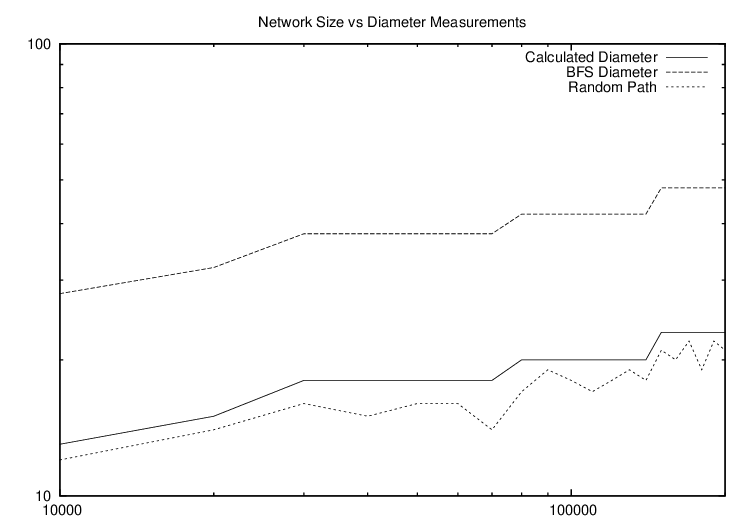}
\caption{Diameter}
\label{fig:figure2}
\end{minipage}
\end{figure}
\begin{figure}[ht]
\begin{minipage}[b]{3in}
\centering
\includegraphics[width=3in]{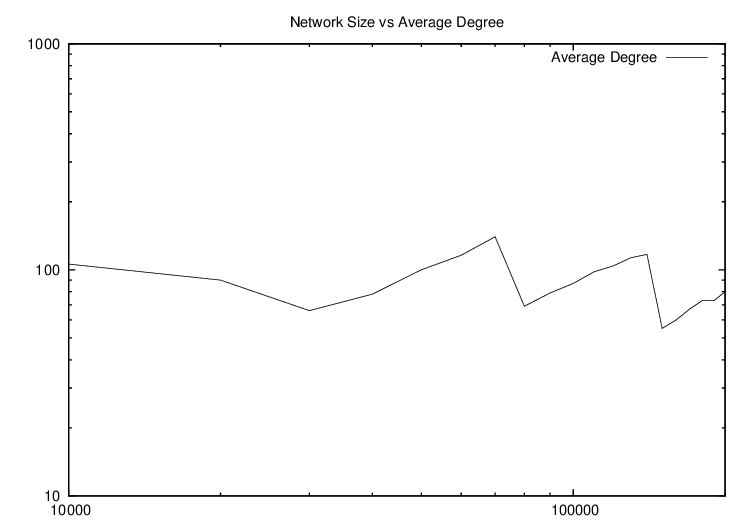}
\caption{Average Degree}
\label{fig:figure3}
\end{minipage}
\hspace{0.5in}\begin{minipage}[b]{3in}
\centering
\includegraphics[width=3in]{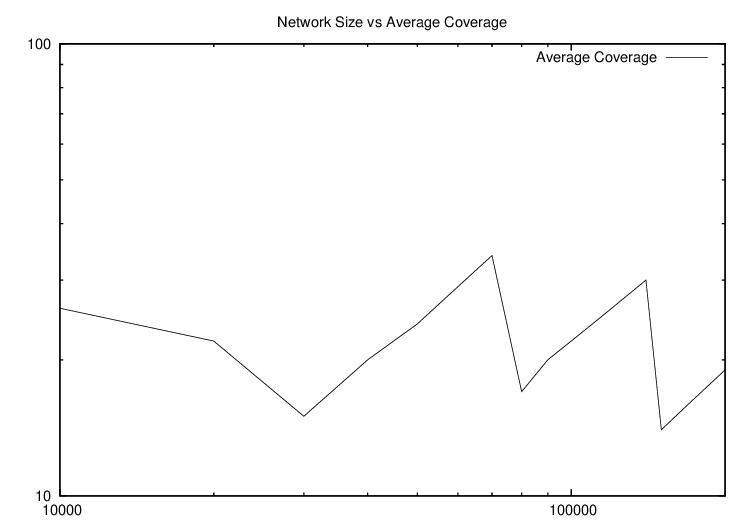}
\caption{Average Coverage}
\label{fig:figure4}
\end{minipage}
\end{figure}

\begin{figure}[h]
\begin{minipage}[b]{3in}
\centering
\includegraphics[width=3in]{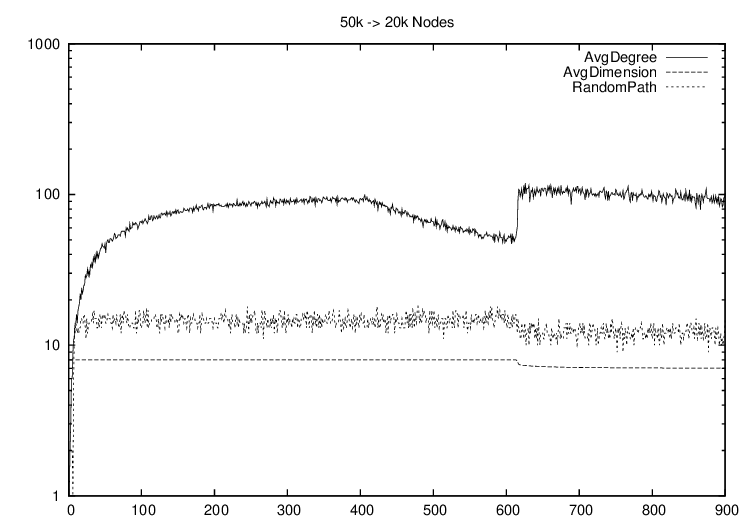}
\caption{Dimension Adjustment}
\label{fig:figure5}
\end{minipage}
\end{figure}

\subsubsection{Coverage}
The coverage (i.e., occupancy of the underlying template graph) of the network is important, without 100\% coverage routing through the network 
cannot be assured to be done efficiently, and if the coverage becomes too low, the network may become
disconnected. Coverage is measured as the percentage of vertices in the template graph (i.e., CCC) that are covered by
a node in the network. Coverage is tied closely to the dimension of the 
CCC graph, in relation to the number of nodes in the network. If the dimension of the CCC graph is too
large for number of nodes, the network can never reach 100\% coverage, as seen in the dimension 10 network
in Figure 1 (all figures are placed in the Appendix). However, it can also be seen that if the dimension fits the number of nodes, the graph
will reach 100\% coverage quickly.
The dimension $r = \lceil\log(N/\log^2 N)\rceil$, with $N$ being stable network size,
 gives 100\% coverage once the network reaches stable size with every simulation.

\subsubsection{Diameter}
\label{stats:diameter}

The diameter $d$ of a CCC graph of dimension $n$ can be computed by $d = 2n + \lfloor n/2\rfloor - 2$ 
for $n\geq4$ \cite{cccdiameter}. The diameter of the network was computed approximately by traversing the 
network with breadth-first search, and taking the diameter to be twice the height of the produced tree.
This provides a reasonable estimate, within a constant factor. With networks of a large number of nodes, 
this becomes to inefficient in terms of time, in many simulations, tripling the run time. A faster approach
is to consider the random path. In the random path, two nodes would be pulled from the network at random, and
a path would be routed between them. $\log(N)$ paths are sampled, and their lengths averaged together. This 
allows a much faster measurement of the network diameter. As seen in Figure 2, the random path actually
provides a much more accurate diameter than the breadth-first search; due to the efficient shape of the CCC graph,
the BFS-diameter is greater by almost a factor of 2.

\subsubsection{Average Degree and Coverage}
\label{stats:avgdegree}

The average degree of a node in the network can be seen in Figure 3 to grow with the network size, keeping
within a constant factor of $\log(N)$. The sharp drops in the average degree occur when the network size is 
large enough to support a higher dimension CCC graph, spreading the nodes in the network over many more 
vertices in the backbone CCC graph. 

The average number of nodes covering a vertex is heavily related to the average degree of a node. As seen
in Figure 4, it follows the same pattern of growth. It is interesting to point out, from networks ranging
of size 10000 to 150000, the average degree and coverage stay constant, around average degree of 100 and
average coverage of 25. 

\subsubsection{Handling Changes in Network Size}
\label{changing}

The network will not stay at constant size forever and must compensate for drastic changes in network size.
To accomplish this, the dimension of the network must be increased or decreased to adjust for an increase or 
decrease in overall network size. This should be accomplished in as decentralized process as possible, so that each
node must work to keep track of the network stability.

We use the following method to detect changes in network size.
Each time unit in the simulator, a node in the network runs a simulation method mimicking normal operations of a
node. At regular intervals, the will look into its neighbors in the network and attempt to ascertain if the current
network is stable, then take measures if it is not stable. The regular intervals were tested with success at 100
to 500 time units; any shorter and the fluctuation of network would interfere too greatly for one individual node
to correctly calculate the network status.

A node determines whether it is stable by using the average degree of several nodes, and 
checking if it is close to the ideal stable degree, $D$. 
The degree of the nodes stays within $O(\log N )$, and based on previous network simulations of nodes in the networks up to 1000000 nodes, the stable degree will fall around 100, $\pm50$, as seen in Figure 3. Each node tracks the progression of 
the average degree of a sampling of nodes in the network, called $A$.
If $A$, begins moving away from the stable degree size, the node will then 
lower its dimension if $A$ is falling or increase its dimension if $A$ is rising. There is a buffer of $\pm65$
around $D$, so that random variations in the sample average degree will not trigger incorrect dimension change.

Each dimension change will only decrease the dimension of the node by 1. This is to prevent the network from 
growing or shrinking too quickly. If the nodes of the network dimensions would make large increases in dimension, 
the nodes would need to expand to too large a CCC, increasing the time the network is disconnected. Decreasing the 
dimension greatly would cause the cycles of the CCC to be shortened too much, causing difficulties in network 
routing.

If each node was left to change by themselves, many nodes would not get a chance to change or change too slowly and 
leave the network unstable or disconnected. To remedy this, once a node detects a network instability and changes dimension, it sends a message to each of its neighbors, suggesting for them to change their dimension to its new 
dimension. Each node monitors these messages, and once it receives enough of them (simulations have shown that 
around 5 is sufficient to eliminate any false positives), it will change its dimension to the suggested dimension, 
regardless of their own measure of the average degree. As each changes, it sends its own suggestion messages, which 
will then effectively propagate the change in dimension across the network. 

As the network is undergoing change, nodes are still joining it, so their dimension is decided by rounding the 
average dimension of all nodes that share its vertex and node id. Since all nodes that share a vertex have the same 
degree, once change to the dimension comes to a vertex, they will all change very quickly, so the new node will either have the new correct dimension or be in a vertex that the change has not propagated to yet.

Figure 5 depicts a typical network response to a large change in network size. The network was simulated for 450000 
time units, where a 50000 node network dropped to a 20000 node network. The network is forgiving in small decreases, 
but once the network drops too far at $t = 610$, the network corrects itself quickly, 70\% of the nodes in the 
network switching in less than 7000 time units. Due to nodes continually being added while the network is adjusting, 
perfect instantaneous convergance to the new dimension is unlikely, but as the network progresses, it will continue  
to self-adjust and reduce the average dimension in all of its nodes to the new correct dimension. The random path 
statistic in Figure 5 shows the average of a sampling of nodes route lengths when trying to reach a random 
assortment of nodes, mimicking requests during normal network operations. While the network contains nodes of 
differing degrees, it is still able to function normally, with few disconnects or routing problems.


\section{Concluding Remarks}
\label{sec:conc}
\vspace{-0.1in}
We presented a simple and general scheme for building a structured P2P network. We analyze our scheme under a realistic churn model and provably show that
it gives essentially optimal bounds with respect to search time, degree,  message complexity and maintenance overhead. The scheme offers  algorithmic benefits to efficient distributed dynamic maintenance of spanning trees.   It will be interesting to explore dynamic algorithms for other problems in this scheme. We also did a simulation based-study the understand the average performance of the scheme
in  networks of moderate size.  For future work, it will be worthwhile to deploy  a P2P system  to evaluate the performance of the proposed scheme.

\bibliographystyle{plain}

\end{document}